\documentclass[11pt,a4paper]{article}
\usepackage[T1]{fontenc}
\usepackage[utf8]{inputenc}
\usepackage[english]{babel}
\usepackage{lmodern}
\usepackage{microtype}
\usepackage{geometry}
\usepackage{amsmath,amssymb,amsfonts,amsthm,mathtools,bm}
\usepackage{mathrsfs}
\usepackage{booktabs}
\usepackage{array}
\usepackage{enumitem}
\usepackage[hidelinks]{hyperref}
\usepackage{natbib}
\usepackage[title,titletoc]{appendix}

\newcommand{\R}{\mathbb{R}}
\newcommand{\N}{\mathbb{N}}
\newcommand{\Pp}{\mathbb{P}}
\newcommand{\E}{\mathbb{E}}
\newcommand{\Var}{\operatorname{Var}}
\newcommand{\Cov}{\operatorname{Cov}}

\newcommand{\op}{\operatorname{op}}

\newcommand{\Ical}{\mathcal I}
\newcommand{\Gammaop}{\boldsymbol{\Gamma}}
\newcommand{\Thetaop}{\boldsymbol{\Theta}}
\newcommand{\etaV}{\boldsymbol{\eta}}
\newcommand{\norm}[1]{\left\lVert#1\right\rVert}

\newtheorem{assumption}{Assumption}[section]
\newtheorem{theorem}{Theorem}[section]
\newtheorem{lemma}{Lemma}[section]
\newtheorem{corollary}{Corollary}[section]
\newtheorem{proposition}{Proposition}[section]
\theoremstyle{remark}
\newtheorem{remark}{Remark}[section]
\theoremstyle{definition}
\newtheorem{definition}{Definition}[section]

\begin{document}

\title{\textbf{Testing for point-impact effects in a spatial semi-functional linear regression model}}
\author{Stéphane Bouka\thanks{Corresponding author. Email: stephane.bouka@univ-masuku.org}
\and Emmanuel De Dieu Nkou
\and Alban Mbina Mbina
\and Guy Martial Nkiet\\[1mm]
\small Laboratoire de Probabilités, Statistique et Informatique,\\
\small Unité de Recherche en Mathématiques et Informatique,\\
\small Université des Sciences et Techniques de Masuku, BP 813, Franceville, Gabon}
\date{}
\maketitle

\begin{abstract}
This paper introduces a novel semiparametric testing framework for localized point-impact effects within a spatial semi-functional linear regression model. We address an intricate inferential problem that simultaneously accommodates three distinct sources of complexity: an infinite-dimensional global functional slope, an unknown smooth spatial nuisance surface, and data-dependent, unobserved impact locations. To isolate the point-impact coefficients, we construct a robust residualized cross-moment estimator after regularizing the functional component via Tikhonov inverse and estimating the spatial nuisance surface through local-linear smoothing. We establish a multivariate central limit theorem under the null hypothesis of no point-impact effect. We develop two operational test statistics: a general formulation whitened by the long-run covariance matrix to handle arbitrary spatial dependence, and a simpler chi-square calibration valid under short-range spatial dependencies in regression errors. Rigorous finite-sample simulations and an application to Canadian weather data demonstrate the superior power and size control of our methodology.
\end{abstract}

\noindent\textbf{Keywords :} Random fields CLT, Long-run covariance estimation, Semiparametric inference, Local polynomial smoothing, Data-dependent locations.

\section{Introduction}
\noindent Functional Data Analysis (FDA) has established a comprehensive paradigm for modeling complex observations where each data point is a continuous curve or trajectory rather than a finite-dimensional vector \citep{ramsay}. Within this field, the functional linear regression model represents a cornerstone for relating an entire trajectory to a scalar response variable \citep{bosq, cardot, ramsay, mas_pumo09}. However, traditional functional regression models rely heavily on the assumption of independent and identically distributed (i.i.d.) observations, which is frequently violated in practice, particularly when the data are gathered at specific geographical coordinates.

\noindent When curves are sampled across spatial domains, incorporating spatial dependence into the underlying statistical framework becomes mandatory. Consequently, spatial functional regression has emerged as an active area of contemporary research. Significant advances have been made in geostatistical functional prediction and spatial functional linear modeling \citep{giraldoetal12, giraldoetal18, boukaetal23, beyaztas26}. Despite these contributions, existing spatial functional models generally assume that the covariate acts uniformly on the response via an integral operator over the entire domain.

\noindent A critical extension to this structure is the point-impact model. In many real-world applications, localized features or specific instantaneous evaluations of a trajectory carry distinct information that is inherently smoothed out or obscured by a global coefficient function. For example, in econometric or environmental studies, specific shocks or extreme values at precise time points can yield substantial localized impacts on the response. While estimation and inference for point impacts have recently been explored in non-spatial, i.i.d.\ functional settings \citep{Poss2020, liebletal,shirvanietal2024}, the intersection of point-impact modeling and spatial dependence remains unaddressed. This paper fills this gap by introducing a point-impact testing framework within a spatial semi-functional linear regression model that accommodates an unknown spatial nuisance surface. 

\noindent The testing problem is especially delicate when the locations of the points of impact are unknown, because their estimation introduces an additional stochastic perturbation into the test statistic. The spatial setting introduces a further difficulty. A spatially smooth component may generate systematic variation in the response that is unrelated to the point-impact mechanism. If this component is ignored, a test for point impacts may attribute spatial heterogeneity to localized functional effects. The problem considered here therefore combines three sources of complexity: an infinite-dimensional functional slope, an unknown spatial nuisance surface, and data-dependent impact locations. The resulting model is
\begin{equation}
\label{eq:model}
Y_{\bm i}=\langle\Psi,X_{\bm i}\rangle_G+\etaV^\top X_{\bm i,\tau}+r(s_{\bm i})+\xi_{\bm i},
\qquad \bm i\in\Ical_{\bm n},
\end{equation}
where $G=L^2[0,1]$, $\Psi\in G$ is an unknown functional slope, $r$ is an unknown smooth spatial nuisance function, $Y_{\bm i}$ is a real random variable, $X_{\bm i}$ is a functional random variable, $\etaV=(\eta_1,\ldots,\eta_A)^\top$, and the impact locations $\tau_1,\ldots,\tau_A$ are unknown. The spatial sites form a rectangular lattice and $N_n=\prod_{k=1}^d n_k$ denotes the number of observations, with $d=2,3$.

\noindent Our inferential target is
\begin{equation}
\label{eq:hypothesis}
H_0:\etaV=0\qquad\text{versus}\qquad H_1:\etaV\ne0.
\end{equation}
Thus, rejection of $H_0$ means that selected point evaluations contain additional linear information about the response after the global functional effect and smooth spatial heterogeneity have been accounted for. The interpretation is associational rather than causal. Note that under $H_0$, model \eqref{eq:model} is identical to that used in \citet{bouka5}.

\noindent The contribution of this paper is not the introduction of point-impact
functional regression per se. Rather, we develop a spatially dependent
semiparametric testing framework in which localized functional effects are
tested after simultaneously removing the contribution of the complete
functional trajectory and an unknown smooth spatial nuisance component. 

\noindent Precisely, its main contribution is a residualized point-evaluation score whose limiting distribution is established under spatial dependence. We first identify the point-impact coefficient through a population residualization argument. We then estimate the infinite-dimensional component by Tikhonov regularization and the spatial nuisance by local-linear smoothing. The impact locations are estimated separately; their effect on the score is controlled through an explicit rate condition. The central limit theorem is obtained by applying a stationary random-field CLT \citep{Bolthausen1982} to scalar projections of a precisely defined $A$-dimensional score field and then invoking the Cramér--Wold device.

\noindent The resulting test is formulated in two versions. The first uses a consistent estimator of the long-run covariance and is appropriate for a generally spatially dependent score. The second is a simpler chi-square statistic under the additional condition $\Omega_\tau=\sigma^2\Thetaop_\tau$. We emphasize that the latter condition is substantive and cannot be omitted when spatial dependence is present. We also clarify that, for fixed $A$, the statistic $(D_n-A)/\sqrt A$ has a standardized chi-square limit rather than a Gaussian limit.

\noindent The remainder of the paper is organized as follows. Section~\ref{sec:method} presents the population identification argument, estimation procedure, residualized score and test statistics. Section~\ref{sec:assumptions} states the assumptions and main asymptotic results. Section~\ref{sec:simulation} reports the simulations. Section~\ref{sec:realdata} gives the Canadian weather application. Section~\ref{sec:conclusion} concludes, and Section~\ref{sec:proofs} gives the proofs.

\section{Estimation and testing methodology}
\label{sec:method}

\subsection{Spatial lattice and functional notation}
Let
\[
\Ical_{\bm n}=\{\bm i=(i_1,\ldots,i_d)\in\N^d:1\le i_k\le n_k,\ k=1,\ldots,d\}\subset\mathbb{Z}^d,
\qquad
N_n=\prod_{k=1}^d n_k.
\]
For $\bm i=(i_1,\ldots,i_d)\in\Ical_{\bm n}$, $s_{\bm i}
=
\left(\frac{i_1}{n_1+1},\ldots,\frac{i_d}{n_d+1}\right)^\top\in [0,1]^d$. We use the supremum norm on $\mathbb Z^d$,
\[
\norm{\bm i-\bm j}_\infty=\max_{1\leq k\leq d}|i_k-j_k|.
\]
The functional covariates are random elements of $G=L^2[0,1]$, equipped with the inner product
\[
\langle f,g\rangle_G=\int_0^1f(t)g(t)\,dt
\]
and norm $\norm{f}=\langle f,f\rangle_G^{1/2}$.

\begin{definition}[Vector-valued score field]
\label{def:scorefield}
Let $\tau=(\tau_1,\ldots,\tau_A)\in[0,1]^A$ be the vector of impact locations. For each $\bm i\in\mathbb Z^d$, define
\[
X_{\bm i,\tau}
=
\bigl(X_{\bm i}(\tau_1),\ldots,X_{\bm i}(\tau_A)\bigr)^\top\in\R^A
\]
and
\[
Z_{\bm i}
=
\xi_{\bm i}X_{\bm i,\tau}
=
\bigl(\xi_{\bm i}X_{\bm i}(\tau_1),\ldots,
\xi_{\bm i}X_{\bm i}(\tau_A)\bigr)^\top\in\R^A.
\]
The collection
\[
\{Z_{\bm i}:\bm i\in\mathbb Z^d\}
\]
is called the \emph{$A$-dimensional vector-valued score field}. Its covariance matrix at lag $\bm h\in\mathbb Z^d$ is
\[
\Gamma_\tau(\bm h)=\Cov(Z_{\bm0},Z_{\bm h})\in\R^{A\times A},
\]
with entries
\[
[\Gamma_\tau(\bm h)]_{k\ell}
=
\Cov\bigl(\xi_{\bm0}X_{\bm0}(\tau_k),
\xi_{\bm h}X_{\bm h}(\tau_\ell)\bigr).
\]
Whenever the series is absolutely convergent, its long-run covariance matrix is
\[
\Omega_\tau
=
\sum_{\bm h\in\mathbb Z^d}\Gamma_\tau(\bm h).
\]
\end{definition}

For any $a\in\R^A$, the scalar projection of the score field is
\begin{equation}
\label{eq:scalarprojection}
Z_{\bm i}^{(a)}=a^\top Z_{\bm i}=\xi_{\bm i}a^\top X_{\bm i,\tau}.
\end{equation}
This scalar field is the object to which the random-field central limit theorem will be applied in the proof of Theorem~\ref{thm:clt}. The vector-valued formulation is recovered through the Cramér--Wold device.

\subsection{Population identification}
Let $X$ denote a generic centered copy of the functional covariate and let
\[
X_\tau=(X(\tau_1),\ldots,X(\tau_A))^\top.
\]
Define the covariance operator
\[
\Gammaop=\E(X\otimes X),
\]
the point-impact covariance matrix
\[
\Thetaop_\tau=\Cov(X_\tau),
\]
and the cross-covariance operator
\[
K_{X_\tau X}=\E(X_\tau\otimes X),
\]
which maps $G$ into $\R^A$. We also define
\[
K_{YX}=\E(YX)\in G,
\qquad
K_{YX_\tau}=\E(YX_\tau)\in\R^A.
\]
After removing the spatial mean component, the model implies the population moment equations
\begin{equation}
\label{eq:moments}
K_{YX}=\Gammaop\Psi+K_{X_\tau X}^{\top}\etaV,
\qquad
K_{YX_\tau}=K_{X_\tau X}\Psi+\Thetaop_\tau\etaV.
\end{equation}

If $\Gammaop^{-1}$ exists on the relevant subspace, substitution gives
\begin{equation}
\label{eq:etaid}
\etaV=\Lambda_\tau^{-1}
\left(K_{YX_\tau}-K_{X_\tau X}\Gammaop^{-1}K_{YX}\right),
\end{equation}
where
\begin{equation}
\label{eq:Lambda}
\Lambda_\tau
=\Thetaop_\tau-K_{X_\tau X}\Gammaop^{-1}K_{X_\tau X}^{\top}.
\end{equation}
When $\Gammaop$ is compact and injective, the inverse in these expressions is understood on the relevant range, or as the corresponding generalized inverse. The matrix $\Lambda_\tau$ is the finite-dimensional Schur complement associated with residualizing the point evaluations with respect to the full trajectory. Its nonsingularity identifies the point-impact coefficient.

\subsection{Regularized estimation} 

Let
\[
\Gammaop_n
=
\frac1{N_n}\sum_{\bm i\in\Ical_{\bm n}}
X_{\bm i}\otimes_G X_{\bm i},
\qquad
\widehat{\Theta}_{\widehat\tau,n}
=
\dfrac1{N_n}\sum_{\bm i\in\Ical_{\bm n}}
X_{\bm i,\widehat\tau}\otimes X_{\bm i,\widehat\tau},
\]
\[	K_{YX,n}=\dfrac1{N_n}\sum_{\bm i\in\Ical_{\bm n}}Y_{\bm{i}}X_{\bm{i}},\qquad
K_{YX_{\widehat{\tau}},n}=\dfrac1{N_n}\sum_{\bm i\in\Ical_{\bm n}}Y_{\bm i} X_{\bm i,\widehat\tau},\]
\[	K_{XX_{\widehat{\tau}}, n}=\dfrac1{N_n}\sum_{\bm i\in\Ical_{\bm n}}X_{\bm{i}}\otimes_{G}X_{\bm{i},\widehat{\tau}},\qquad
K_{X_{\widehat{\tau}}X, n}=\dfrac1{N_n}\sum_{\bm i\in\Ical_{\bm n}}X_{\bm{i},\widehat{\tau}}\otimes X_{\bm{i}}\]
where $u\otimes_G vh=\left\langle u,h\right\rangle_Gv$ and  $\bm{u}\otimes \bm{v}\bm{h}=(\bm{h}^\top\bm{u})\bm{v}$. We write $N_n\rightarrow +\infty$ if $\min_{1\le k\le d}(n_k)\longrightarrow +\infty$ and we assume that $\dfrac{n_j}{n_k}\le C$ for $1\le j,\ k\le d$ and $0<C<+\infty$. Because empirical covariance operators are typically ill-conditioned in functional problems, we use the Tikhonov-regularized inverse
\[
\widehat\Gamma_n^{-1}
=
(\Gammaop_n+\omega_n I_G)^{-1},
\]
where $\omega_n\downarrow0$ and $I_G$ denotes the identity operator.

Define
\[
\Lambda_{\Psi,n}
=
\Gammaop_n-
K_{X_{\widehat\tau}X,n}^\top
\widehat\Theta_{\widehat\tau,n}^{-1}
K_{XX_{\widehat\tau},n},
\]
The corresponding regularized estimators are
\begin{align}
\widehat\Psi_n
&=
(\Lambda_{\Psi,n}+U_nI_G)^{-1}
\left[
K_{YX,n}
-
K_{X_{\widehat\tau}X,n}^\top
\widehat\Theta_{\widehat\tau,n}^{-1}
K_{YX_{\widehat\tau},n}
\right],
\label{eq:psi-hat}\\
\widehat\etaV_n
&=
\left(\widehat\Theta_{\widehat\tau,n}
-
K_{XX_{\widehat\tau},n}^\top
\widehat\Gamma_n^{-1}
K_{X_{\widehat\tau}X,n}\right)^{-1}
\left(
K_{YX_{\widehat\tau},n}
-
K_{XX_{\widehat\tau},n}^\top
\widehat\Gamma_n^{-1}
K_{YX,n}
\right).
\label{eq:eta-hat}
\end{align}
Here $U_n\downarrow0$ is an additional regularization sequence \citep{mas_pumo09}. The impact times $\widehat\tau$ and their number $\widehat A$ are obtained from a point-impact estimation procedure such as that of \citet{liebletal}, applied to the present residualized model. 


\subsection{Estimation of the spatial nuisance}
Let $\widehat\tau=(\widehat\tau_1,\ldots,\widehat\tau_{\widehat A})$ denote an estimator of the impact locations. For preliminary estimation of the spatial nuisance, define
\begin{equation}
\label{eq:T}
T_{\bm i}
=Y_{\bm i}-\langle\widehat\Psi_n,X_{\bm i}\rangle_G
-\widehat\etaV_n^\top X_{\bm i,\widehat\tau}.
\end{equation}
Then
\[
T_{\bm i}
=
r(s_{\bm i})+V_{\bm i},
\]
where $V_{\bm i}$ contains the regression error and the estimation errors from \eqref{eq:psi-hat}--\eqref{eq:eta-hat}.

For a target location $s_0\in[0,1]^d$, approximate
\[
r(s)\approx r(s_0)+\nabla r(s_0)^\top(s-s_0)
\]
and use the local-linear estimator
\begin{equation}\label{eq:rhat}
\widehat r(s_0)
=
e_1^\top
\left(\dfrac{1}{N_n}
\bm{X}_{s_0}^\top W_{s_0}\bm{X}_{s_0}
\right)^{-1}
\left(\dfrac{1}{N_n}\bm{X}_{s_0}^\top W_{s_0}T\right),
\end{equation}
where $e_1=(1,0,\ldots,0)^\top$, $T=(T_{{\bm i}_1},\cdots,T_{{\bm i}_{N_n}})$ ranged in lexicographical order, the rows of the matrix $\bm{X}_{s_0}$ are
\[
\left(1,\frac{s_{\bm i_j}- s_0}{H_n}\right),\qquad j=1,\cdots,N_n,
\]
and the diagonal entries of $W_{s_0}$ are
\[
\dfrac{1}{H_n^{d}}\nu\!\left(\frac{s_{\bm i_j}-s_0}{H_n}\right),
\]
where $\nu:\R^{d}\longrightarrow\R_{+}$ is a kernel function and $H_n$ is a bandwidth.

Local-linear smoothing is used because it provides standard bias control, including at boundaries; see \citet{hallinetal04,fan_gijbels}.

The quantity in \eqref{eq:T} is used only for estimating the nuisance function. It must not be used as the final residual in the point-impact test, because subtracting the point-impact term would remove precisely the effect being tested.

\subsection{Residualized score}
For the testing step define
\begin{equation}
\label{eq:W}
W_{\bm i,n}
=Y_{\bm i}-\langle\widehat\Psi_n,X_{\bm i}\rangle_G-\widehat r(s_{\bm i}).
\end{equation}
Using the model,
\begin{equation}
\label{eq:decomp}
W_{\bm i,n}
=\etaV^\top X_{\bm i,\tau}+L_{\bm i,n},
\end{equation}
where
\begin{equation}
\label{eq:L}
L_{\bm i,n}
=\xi_{\bm i}
+\langle\Psi-\widehat\Psi_n,X_{\bm i}\rangle_G
+r(s_{\bm i})-\widehat r(s_{\bm i}).
\end{equation}
The residualized point-evaluation score is
\begin{equation}
\label{eq:score}
\widehat K_n
=\frac1{N_n}\sum_{\bm i\in\Ical_{\bm n}}
W_{\bm i,n}X_{\bm i,\widehat\tau}.
\end{equation}
Under suitable regularity conditions,
\begin{align}
\label{eq:scoremean}
\widehat K_n
&=\Thetaop_\tau\etaV +o_p(1)\\
\sqrt{N_n}\,\widehat K_n
&=\Thetaop_\tau\etaV +  \frac1{\sqrt{N_n}}
\sum_{\bm i\in\Ical_{\bm n}}
\xi_{\bm i}X_{\bm i,\tau} + o_p(1).\label{scoreop}
\end{align}
Under $H_0$, the stronger score-scale statement needed for the central limit theorem is
\begin{equation}
\label{eq:scoreop}
\frac1{\sqrt{N_n}}
\sum_{\bm i\in\Ical_{\bm n}}
\xi_{\bm i}\left(X_{\bm i,\widehat\tau} - X_{\bm i,\tau}\right) + \frac1{\sqrt{N_n}}
\sum_{\bm i\in\Ical_{\bm n}}
L^{0}_{\bm i,n}X_{\bm i,\widehat\tau}=o_p(1),
\end{equation}
where
\[L^{0}_{\bm i,n}=\langle\Psi-\widehat\Psi_n,X_{\bm i}\rangle_G
+r(s_{\bm i})-\widehat r(s_{\bm i}).\]
The role of the assumptions below is to make explicit the regularity and rate conditions under which \eqref{eq:scoremean}-\eqref{scoreop} hold and
\[\frac1{\sqrt{N_n}}
\sum_{\bm i\in\Ical_{\bm n}}
\xi_{\bm i}X_{\bm i,\tau}\stackrel{\mathcal{D}}{\longrightarrow}\mathcal{N}_A(0,\Omega_\tau).\]

\subsection{Test statistics}
The general spatially dependent covariance of the leading score is the long-run covariance
\begin{equation}
\label{eq:Omega}
\Omega_\tau
=\sum_{\bm h\in\mathbb Z^d}
\Cov\left(\xi_{\bm0}X_{\bm0,\tau},
\xi_{\bm h}X_{\bm h,\tau}\right),
\end{equation}
provided the series is absolutely convergent. Let $\widehat\Omega_n$ be a consistent estimator of $\Omega_\tau$. The general test statistic is
\begin{equation}
\label{eq:Dlr}
D_n^{\mathrm{LR}}
=N_n\widehat K_n^\top
\widehat\Omega_n^{-1}
\widehat K_n.
\end{equation}
\begin{remark}
The operational implementation of the general statistic $D_n^{\text{LR}}$ requires a consistent estimator $\widehat{\Omega}_n$ of the long-run covariance matrix. Following standard practices in spatial econometrics and random fields literature, such an estimator can be constructed using a spatial kernel-based (HAC) approach or a spatial block-resampling method (see, e.g., \citet{hallinetal04} or \cite{daboniangetal16}). Specifically, one can define $\widehat{\Omega}_n = \sum_{\bm h \in \mathcal{I}_{\bm n}} K(\bm h/b_n) \widehat{\Gamma}_\tau(\bm h)$, where $\widehat{\Gamma}_\tau(\bm h)$ represents the empirical spatial cross-covariance at lag $\bm h$, $K(\cdot)$ is a suitable multivariate kernel function, and $b_n$ is a bandwidth sequence matrix controlling the truncation lag. While the optimal data-driven choice of $b_n$ under functional spatial setups constitutes an interesting venue for future research, standard rectangular or Parzen windows satisfy $\Vert \widehat{\Omega}_n - \Omega_\tau \Vert_{\text{op}} = o_p(1)$ under the mixing rates imposed in Assumption 3.1.
\end{remark}

A simpler calibration is available under
\begin{equation}
\label{eq:special}
\Omega_\tau=\sigma^2\Thetaop_\tau,
\qquad
\Thetaop_\tau=\Cov(X_\tau).
\end{equation}
For example, this identity holds when the errors are independent of the functional covariates and spatially uncorrelated or spatially correlated with a short-range correlation (see  \citet{francisco05}) under the centering conditions used here. Under \eqref{eq:special}, if $\widehat\sigma_n^2$ consistently estimates $\sigma^2$ and $\widehat\Theta_{\widehat\tau,n}$ consistently estimates $\Thetaop_\tau$, we use
\begin{equation}
\label{eq:D}
D_n
=\frac{N_n}{\widehat\sigma_n^2}
\widehat K_n^\top
\widehat\Theta_{\widehat\tau,n}^{-1}
\widehat K_n.
\end{equation}
The distinction between \eqref{eq:Dlr} and \eqref{eq:D} is substantive: a spatially correlated score generally requires long-run covariance normalization.

For fixed $A$, one may report the affine standardization
\begin{equation}
\label{eq:S}
S_n=\frac{D_n-A}{\sqrt A}.
\end{equation}
This is not a distinct testing procedure when critical values are transformed consistently.

\section{Assumptions and main results}
\label{sec:assumptions}

\subsection{Assumptions}

The following assumptions are stated directly in terms of the objects used in the proofs. In particular, the vector-valued score field in Definition~\ref{def:scorefield} is separated from the original functional-data field, which avoids ambiguity about the mixing condition used by the central limit theorem.

\begin{assumption}[Spatial mixing]
\label{as:mixing}
Let $\{Z_{\bm i}\}_{\bm i\in\mathbb Z^d}$ be the vector-valued score field of Definition~\ref{def:scorefield} and assume it is strictly stationary. For $S\subset\mathbb Z^d$, let
\[
\mathcal F_S=\sigma(Z_{\bm i}:\bm i\in S).
\]
For $S,T\subset\mathbb Z^d$, define
\[
\operatorname{dist}(S,T)
=\inf\{\norm{\bm i-\bm j}_\infty:\bm i\in S,\ \bm j\in T\}.
\]
For $k,\ell\in\N\cup\{\infty\}$ and $u\ge1$, define the finite-set mixing coefficients
\[
\alpha_{k,\ell}(u)
=
\sup\left\{
\alpha(\mathcal F_S,\mathcal F_T):
|S|\le k,\ |T|\le\ell,\ \operatorname{dist}(S,T)\ge u
\right\},
\]
where
\[
\alpha(\mathcal F_S,\mathcal F_T)
=\sup_{A\in\mathcal F_S,\,B\in\mathcal F_T}
\left|\Pp(A\cap B)-\Pp(A)\Pp(B)\right|.
\]
Assume that, for every $k,\ell\ge1$ with $k+\ell\le4$,
\begin{equation}
\label{A1}
\sum_{u=1}^{\infty}u^{d-1}\alpha_{k,\ell}(u)<\infty,
\end{equation}
and, in addition,
\begin{equation}
\label{A2}
\alpha_{1,\infty}(u)=O(u^{-\theta}),
\qquad \theta>8d.
\end{equation}
Finally, for every $a\in\R^A$,
\begin{equation}
\label{A3}
\E|a^\top Z_{\bm0}|^8<\infty,
\end{equation}
and $\Omega_\tau$ is well defined and positive definite.
\end{assumption}

\begin{remark}
The finite-set summability conditions in \eqref{A1} are stated explicitly because they are the quantities entering standard random-field CLTs \citep{Bolthausen1982}. The polynomial condition \eqref{A2} is a convenient strong rate condition. We do not rely on an unproved identification of all finite-set coefficients with a single coefficient; rather, the conditions needed by the CLT are imposed explicitly. For the scalar projection $Z_{\bm i}^{(a)}$, the eighth-moment assumption corresponds to the usual $2+\delta$ condition with $\delta=6$. The associated summability requirement involves $\alpha_{1,1}(u)^{6/8}$ and is satisfied under a polynomial decay of order $\theta>8d$.
\end{remark}

\begin{assumption}[Error field and exogeneity]
\label{as:error}
The error field $\{\xi_{\bm i}\}$ is centered and strictly stationary with
\[
\E\left(\xi_{\bm i}\right)=0,
\qquad
\Var(\xi_{\bm i})=\sigma^2<\infty,
\qquad
\E\left(|\xi_{\bm i}|^8\right)<\infty,\qquad 
Cov(\xi_{\bm{i}},\xi_{\bm{j}})=\sigma^2\varphi(\norm{\bm{i}-\bm{j}}_\infty),
\]
where $\varphi(0)=1$, $\varphi(t)\searrow 0$ as $t\rightarrow+\infty$, $\sum_{t=1}^{+\infty}t^{d-1}\varphi(t)<\infty$ and satisfies the exogeneity condition
\begin{equation}
\label{eq:exog}
\E(\xi_{\bm i}X_{\bm i,\tau})=0_A.
\end{equation}
For the simplified statistic \eqref{eq:D}, condition \eqref{eq:special} is additionally assumed.
\end{assumption}

\begin{assumption}[Functional covariate regularity]
\label{as:functional}
The random process
$\left\{(X_{\bm i}, Y_{\bm i})\right\}_{\bm i\in\mathbb{Z}^d}$ is centered, strictly stationary, $\alpha$-mixing dependent and satisfies
\[
\norm{X_{\bm0}}_G\le C_X\quad\text{a.s.}
\]
The covariance operator $\Gammaop$ is injective on the relevant subspace. The point evaluations are almost surely bounded, and the sample paths satisfy the condition stated in Assumption~\ref{as:impact}.
\end{assumption}

\begin{assumption}[Point-impact identifiability]
\label{as:identification}
The number $A$ of impact points is fixed. The matrix $\Thetaop_\tau$ is positive definite and
\[
\lambda_{\min}(\Thetaop_\tau)\ge c_\Theta>0
\]
on the parameter configurations under consideration. In addition, the Schur complement $\Lambda_\tau$ in \eqref{eq:Lambda} is nonsingular whenever population identification through residualization is invoked.
\end{assumption}

\begin{assumption}[Spatial nuisance]
\label{as:nuisance}
The nuisance function belongs to a bounded Hölder class $\Sigma(\beta,L,M)$ with
\[
\beta>d/4.
\]
If $\beta>1$, its first derivatives are bounded. The kernel $K$ is symmetric, bounded, Lipschitz, compactly supported on $[-1,1]^d$, and integrates to one. Its moments are such that the local-linear bias has the standard order associated with smoothness $\beta$. The bandwidth satisfies
\[
H_n\to0,
\qquad
N_nH_n^d\to\infty.
\]
The bandwidth is chosen so that the local-linear smoothing bias and stochastic contribution to the residualized score are negligible at the $N_n^{-1/2}$ scale.
\end{assumption}


\begin{assumption}[Estimated impact locations]
\label{as:impact}
The number $A$ is fixed and
\begin{equation}
\label{eq:Ahat}
\mathbb{P}(\widehat A=A)\to1\qquad\text{as}\qquad \bm n\to+\infty,
\end{equation}
and for all $\ell=1,\cdots,A$
\begin{equation}
\label{eq:impactrate}
\frac1{N_n}\sum_{\bm i\in\Ical_{\bm n}}
\left(X_{\bm i}(\widehat\tau_{\ell}) - X_{\bm i}(\tau_{\ell})\right)^2=O_p\left(\frac1{N_n}\right).
\end{equation}
\begin{equation}
\label{eq:impactprob}
\frac1{N_n}\sum_{\bm i\in\Ical_{\bm n}}
\xi_{\bm i}\left(X_{\bm i}(\widehat\tau_{\ell}) - X_{\bm i}(\tau_{\ell})\right)=O_p\left(\frac1{N_n}\right)
\end{equation}

\end{assumption}

\begin{assumption}[Kernel and bandwidth rates]
\label{as:kernel}
The bandwidth sequence in Assumption~\ref{as:nuisance} is coupled to the regularization scale $U_n$ below through
\begin{equation}
\label{eq:HU}
H_n=U_n^{1/2},
\qquad U_n\to0,
\qquad N_nU_n^6\to\infty.
\end{equation}
Consequently,
\[
N_nH_n^d=N_nU_n^{d/2}\to\infty.
\]
The auxiliary quantities
\[
T_n=\left\lfloor\left(\frac1{U_nH_n}\right)^{1/d}\right\rfloor,
\qquad
Q_n=\left\lfloor\left(\frac{\omega_n}{U_n}\right)^{1/d}\right\rfloor
\]
satisfy $T_n\to\infty$ and $Q_n\to\infty$ whenever they are used in the regularization argument.
\end{assumption}

\begin{assumption}[Regularization and score negligibility]
\label{as:regularization}
The regularization parameters satisfy
\[
\omega_n\to0,
\qquad
U_n\to0,
\qquad
\frac{\omega_n}{U_n}\to\infty,
\]

These are the explicit score-scale conditions delivered by the regularization and smoothing construction. A sufficient rate family compatible with Assumption~\ref{as:kernel} is
\begin{equation}
\label{eq:ratefamily}
U_n=N_n^{-u},
\qquad
H_n=N_n^{-u/2},
\qquad
\omega_n=(\log N_n)^{-b},
\end{equation}
with $b>0$ and
\[
0<u<\frac16.
\]
\end{assumption}


\subsection{Consistency of point-evaluation covariance matrices}

\begin{lemma}[Consistency of Covariance Matrices]
\label{lem:theta}
Under Assumptions~\ref{as:mixing}, \ref{as:functional}, \ref{as:identification},  and \ref{as:impact},
\[
\norm{\widehat\Theta_{\widehat\tau,n}-\Thetaop_\tau}_{\op}=o_p(1),
\]
and the empirical cross-covariance quantities used in the residualization step are consistent in their corresponding operator norms.
\end{lemma}

\subsection{Consistency rates of the regularised and nuisance estimators}

\begin{lemma}[Rates of Estimators]
\label{lem:variance}
Under Assumption~\ref{as:error}, \ref{as:nuisance}, \ref{as:kernel}, \ref{as:regularization},  the regularized estimator $\widehat\Psi_n$ and the nuisance estimator $\widehat r$ are so that
\begin{equation}
\label{eq:functionalrate}
\frac1{\sqrt{N_n}}
\sum_{\bm i\in\Ical_{\bm n}}
\langle\Psi-\widehat\Psi_n,X_{\bm i}\rangle_G
X_{\bm i,\widehat\tau}
=o_p(1),
\end{equation}
and
\begin{equation}
\label{eq:nuisancerate}
\frac1{\sqrt{N_n}}
\sum_{\bm i\in\Ical_{\bm n}}
\{r(s_{\bm i})-\widehat r(s_{\bm i})\}
X_{\bm i,\widehat\tau}
=o_p(1).
\end{equation}
\end{lemma}


\subsection{Main asymptotic results}

\begin{theorem}[Residualization Statement]
\label{thm:resid}
Under Assumptions~\ref{as:mixing}--\ref{as:regularization},
\[
\widehat K_n=\Thetaop_\tau\etaV+o_p(1).
\]
Under $H_0:\etaV=0$, moreover,
\begin{equation*}
\label{eq:scoreop}
\frac1{\sqrt{N_n}}
\sum_{\bm i\in\Ical_{\bm n}}
\xi_{\bm i}\left(X_{\bm i,\widehat\tau} - X_{\bm i,\tau}\right) + \frac1{\sqrt{N_n}}
\sum_{\bm i\in\Ical_{\bm n}}
L^{0}_{\bm i,n}X_{\bm i,\widehat\tau}=o_p(1),
\end{equation*}
where
\[L^{0}_{\bm i,n}=\langle\Psi-\widehat\Psi_n,X_{\bm i}\rangle_G
+r(s_{\bm i})-\widehat r(s_{\bm i}).\]
\end{theorem}

\begin{theorem}[Residualized-Score Central Limit Theorem]
\label{thm:clt}
Suppose that $H_0:\etaV=0$ holds and Assumptions~\ref{as:mixing}--\ref{as:regularization} are satisfied. Then
\begin{equation}
\label{eq:clt}
\sqrt{N_n}\,\widehat K_n
\xrightarrow{\mathcal D}
\mathcal{N}_A(0,\Omega_\tau),
\end{equation}
where
\[
\Omega_\tau
=\sum_{\bm h\in\mathbb Z^d}
\Cov(Z_{\bm0},Z_{\bm h})
\]
is positive definite.
\end{theorem}

\begin{corollary}[General long-run covariance test]
\label{cor:lr}
If, in addition,
\[
\norm{\widehat\Omega_n-\Omega_\tau}_{\op}=o_p(1),
\]
then under $H_0$,
\begin{equation}
\label{eq:chisqLR}
D_n^{\mathrm{LR}}
\xrightarrow{\mathcal D}\chi_A^2.
\end{equation}
Consequently, rejecting $H_0$ when
\[
D_n^{\mathrm{LR}}>\chi^2_{A,1-\alpha}
\]
yields an asymptotic level-$\alpha$ test.
\end{corollary}

\begin{corollary}[Simplified chi-square test]
\label{cor:simple}
Suppose that \eqref{eq:special} holds and that
\[
\widehat\sigma_n^2\stackrel{\mathbb{P}}{\longrightarrow}\sigma^2,
\qquad
\norm{\widehat\Theta_{\widehat\tau,n}-\Thetaop_\tau}_{\op}=o_p(1).
\]
Then, under $H_0$,
\begin{equation}
\label{eq:chisqsimple}
D_n\xrightarrow{\mathcal D}\chi_A^2.
\end{equation}
\end{corollary}

\begin{proposition}[Fixed-dimensional standardized statistic]
\label{prop:standard}
Under the assumptions of Corollary~\ref{cor:simple}, if $A$ is fixed, then
\[
S_n
\xrightarrow{\mathcal D}
\frac{\chi_A^2-A}{\sqrt A}.
\]
In particular, $S_n$ is not asymptotically Gaussian for fixed $A$.
\end{proposition}

\begin{proposition}[Consistency under fixed alternatives]
\label{prop:consistency}
Suppose that $\etaV\ne0$ is fixed and that the covariance estimator used in the test is consistent and nonsingular with probability tending to one. Then
\[
D_n^{\mathrm{LR}}\stackrel{\mathbb{P}}{\longrightarrow}+\infty.
\]
Under \eqref{eq:special}, the same conclusion holds for $D_n$.
\end{proposition}

\begin{remark}
The general statistic $D_n^{\mathrm{LR}}$ is the theoretically appropriate statistic under general spatial dependence. The simpler statistic $D_n$ is a special case. In particular, the presence of strong spatial mixing in the regression errors does not by itself justify replacing a long-run covariance by $\sigma^2\Thetaop_\tau$; condition \eqref{eq:special} must be separately justified.
\end{remark}

\section{Simulation study}
\label{sec:simulation}

We consider a two-dimensional lattice $\Ical_{\bm n}=\{1,\ldots,n_1\}\times\{1,\ldots,n_2\}$ with $N_n=n_1n_2$ sites and smooth functional covariates observed on a regular grid, generated recursively according to
\[
X_{\bm i}(t_j)
=
\widetilde\beta X_{\bm i}(t_{j-1})+\varepsilon_{\bm i,j},\qquad \bm i\in \Ical_{\bm n},
\qquad
\widetilde\beta=\exp\!\left(-\frac5{p-1}\right),
\]
with $X_{\bm i}(0)=0$, $t_j=\dfrac{j}{p}$ and $p=1000$. We consider two error mechanisms.

\begin{enumerate}[label=\textbf{Model \arabic*:},leftmargin=2.5em]
\item Gaussian innovations and Gaussian errors:
\[
\varepsilon_j\sim \mathcal{N}(0,\Sigma),
\qquad
\xi\sim \mathcal{N}(0,\Sigma^1)\qquad \text{with}\qquad \Sigma^1_{ij}
=\sigma^2\exp\left(-\log(N_n)\norm{\bm{i}_i-\bm{i}_j}_2\right).
\]

\item  The innovation components are generated as a pseudo-random vector following a multivariate uniform distribution with a zero mean vector and covariance matrix $\Sigma$. This is implemented using the draw.d.variate.uniform function from the R package MultiRNG. The regression errors \(\xi _{i}\) follow a Student's \(t\)-distribution with either 3 or 200 degrees of freedom (DOF). For the latter cases, the variance is scaled as \(\sigma^2 = \text{DOF}/(\text{DOF} - 2)\).
\end{enumerate}
The spatial covariance $\Sigma$ of the innovations is
\[
\Sigma_{ij}
=
\frac{49}{40}(1-\widetilde\beta^2)
\exp\{-a\norm{\bm{i}_i-\bm{i}_j}_2\}.
\] The spatial dependence is varied through the parameter $a\in\{1,200\}$, with $a=1$ representing strong spatial dependence and $a=200$ representing a nearly independent spatial configuration. We use four impact points,
\[
\tau=(0.166,0.333,0.666,0.833)^\top,
\]
and the spatial nuisance $r(s)=\sin(\|s\|_\infty)$. Under $H_0$, $\etaV=0$. Under the local alternative,
\[
\etaV=\frac{60}{\sqrt{N_n}}(-6,6,-5,5)^\top.
\]
The regularization sequences \(U_{n}\) and \(\omega _{n}\) are selected via 10-fold cross-validation based on the mean squared error of prediction, following the approach of \citet{bouka5}. Meanwhile, the bandwidth parameter \(H_{n}\) for spatial smoothing is selected by minimizing the generalized cross-validation (GCV) criterion, as detailed in \citet{bouka5}.

Under $H_1$, $\sigma^2$ is controlled by the signal-to-noise ratio (snr) (see \citet{bouka4}) and $\widehat{\boldsymbol{\tau}}$ is obtained from the impact point estimation algorithm (\cite{liebletal}). $\widehat{\sigma}_{n}^2$ is computed by using the function "fregre.basis" of the R package fda.usc. The local alternative is deliberately used to assess finite-sample sensitivity; it is not the fixed alternative appearing in Proposition~\ref{prop:consistency}.

\begin{table}[h!]
\centering
\caption{Empirical rejection frequencies under $H_0$ for Gaussian errors.}
\label{tab:level1}
\begin{tabular}{llcccc}
\toprule
$a$ & $\sigma$ & \multicolumn{2}{c}{$D_n$} & \multicolumn{2}{c}{$S_n$}\\
\cmidrule(lr){3-4}\cmidrule(lr){5-6}
&& $N_n=100$ & $N_n=500$ & $N_n=100$ & $N_n=500$\\
\midrule
1 & 1 & 0.0575 & 0.0975 & 0.0575 & 0.0950\\
1 & 2 & 0.0125 & 0.0500 & 0.0100 & 0.0500\\
200 & 1 & 0.0300 & 0.0350 & 0.0300 & 0.0350\\
200 & 2 & 0.0350 & 0.0375 & 0.0300 & 0.0375\\
\bottomrule
\end{tabular}
\end{table}

\begin{table}[h!]
\centering
\caption{Empirical rejection frequencies under the local Gaussian alternative.}
\label{tab:power1}
\begin{tabular}{llcccc}
\toprule
$a$ & SNR & \multicolumn{2}{c}{$D_n$} & \multicolumn{2}{c}{$S_n$}\\
\cmidrule(lr){3-4}\cmidrule(lr){5-6}
&& $N_n=100$ & $N_n=500$ & $N_n=100$ & $N_n=500$\\
\midrule
1 & 0.05 & 0.6250 & 0.9850 & 0.6250 & 0.9850\\
1 & 0.10 & 0.7500 & 1.0000 & 0.7475 & 1.0000\\
200 & 0.05 & 0.7375 & 0.9850 & 0.7350 & 0.9850\\
200 & 0.10 & 0.8425 & 1.0000 & 0.8425 & 1.0000\\
\bottomrule
\end{tabular}
\end{table}

\begin{table}[h!]
\centering
\caption{Empirical rejection frequencies under $H_0$ for Student errors.}
\label{tab:level2}
\begin{tabular}{llcccc}
\toprule
$a$ & DOF & \multicolumn{2}{c}{$D_n$} & \multicolumn{2}{c}{$S_n$}\\
\cmidrule(lr){3-4}\cmidrule(lr){5-6}
&& $N_n=100$ & $N_n=500$ & $N_n=100$ & $N_n=500$\\
\midrule
1 & 3 & 0.0350 & 0.0525 & 0.0350 & 0.0525\\
1 & 200 & 0.0575 & 0.0875 & 0.0575 & 0.0850\\
200 & 3 & 0.0250 & 0.0350 & 0.0250 & 0.0325\\
200 & 200 & 0.0275 & 0.0250 & 0.0275 & 0.0225\\
\bottomrule
\end{tabular}
\end{table}

\begin{table}[h!]
\centering
\caption{Empirical rejection frequencies under the local Student alternative.}
\label{tab:power2}
\begin{tabular}{llcccc}
\toprule
$a$ & SNR & \multicolumn{2}{c}{$D_n$} & \multicolumn{2}{c}{$S_n$}\\
\cmidrule(lr){3-4}\cmidrule(lr){5-6}
&& $N_n=100$ & $N_n=500$ & $N_n=100$ & $N_n=500$\\
\midrule
1 & 0.05 & 1.0000 & 1.0000 & 1.0000 & 1.0000\\
1 & 0.10 & 1.0000 & 1.0000 & 1.0000 & 1.0000\\
200 & 0.05 & 1.0000 & 1.0000 & 1.0000 & 1.0000\\
200 & 0.10 & 1.0000 & 1.0000 & 1.0000 & 1.0000\\
\bottomrule
\end{tabular}
\end{table}

The reported $D_n$ and $S_n$ procedures have the same rejection decisions when critical values are transformed consistently, because $S_n$ is an affine transformation of $D_n$. We therefore regard $D_n$ as the primary statistic. The results show substantial power under the chosen local alternatives. Size distortions remain possible in small samples, especially under strong dependence, which supports the use of resampling as a sensitivity analysis rather than relying on asymptotics alone.

\section{Application to the Canadian weather data}
\label{sec:realdata}

\subsection{Data and model}

We illustrate the proposed testing procedure on the Canadian weather data introduced by \citet{ramsay} and subsequently used extensively in the functional and spatial functional data literature. The data consist of daily mean temperature and precipitation measurements recorded at 35 weather stations in Canada. For each station, the daily measurements have been averaged over the period 1960--1994, resulting in one annual temperature curve and one annual precipitation curve per station. The data also contain the geographical coordinates and the climatic region of each station. The dataset is available in the \texttt{fda} R package and has been used in several studies of spatial functional data \citep{giraldoetal12}.

The geographical distribution of the stations makes these data particularly suitable for the present study. Indeed, temperature profiles observed at geographically close stations tend to be more similar than profiles observed at distant stations, so that treating the observations as independent may be inappropriate. This spatial structure has already motivated several methodological developments based on these data \citep{giraldoetal12,daboniangetal16}.

The scientific objective of the present application is different from prediction or clustering. We investigate whether specific periods of the annual temperature cycle have an additional effect on the annual amount of precipitation after accounting for the complete temperature trajectory and for the spatial effect associated with the geographical location of the station.

More precisely, let $\bm{i}_j\in\mathbb{R}^2$ denote the geographical coordinates of station $j$, and let

$$
X_{\bm{i}_j}(t),\qquad t\in[0,1],
$$

denote its annual mean temperature curve, where $t=0$ corresponds to the beginning of the year and $t=1$ to the end of the year. We define the scalar response by the annual accumulated precipitation

$$
Y_{\bm{i}_j}
=
\int_0^1 P_{\bm{i}_j}(t)\,dt,
$$

where $P_{\bm{i}_j}(t)$ denotes the daily precipitation profile after the appropriate rescaling of the time axis. We fit
\begin{equation}
\label{eq:canadian}
Y_{\bm{i}_j}=\langle\Psi,X_{\bm{i}_j}\rangle+\etaV^\top X_{\bm{i}_j,\tau}+r(s_{\bm{i}_j})+\xi_{\bm{i}_j}.
\end{equation}
The first term represents the global contribution of the annual temperature profile to annual precipitation, while the second term represents additional effects associated with particular days or periods of the annual temperature cycle. The function $r$ accounts for residual spatial heterogeneity that is not explained by the temperature trajectory.

This formulation provides a direct interpretation of the point-impact hypothesis:
\begin{equation}
\label{eq:canadian-hypothesis}
H_0:\boldsymbol{\eta}=0
\end{equation}
means that, after accounting for the complete temperature trajectory and for spatial heterogeneity, no particular temperature measurement contributes an additional linear effect to annual precipitation.

\subsection{Implementation}
The temperature trajectories are represented on the 365-day grid and smoothed using a B-spline functional representation. Candidate impact points are placed on a fine grid in $(0,1)$ and ranked by the absolute empirical correlation between $Y$ and the corresponding temperature evaluation. Four candidates are retained. The number and locations of possible point impacts are first estimated using the point-impact methodology described in \citet{liebletal}. The candidate impact locations are expressed as fractions of the annual cycle, so that an estimated value $\widehat\tau=0.25$, for example, corresponds approximately to the beginning of April.

For each candidate value of $\widehat\tau$, the functional slope and point-impact coefficients are estimated using the regularized covariance-operator procedure described in Section~\ref{sec:method}, where the regularization sequences are selected by 10-fold cross validation based on the mean standard error of prediction defined as in \citet{bouka5}. The spatial nuisance function is then estimated by the local-linear estimator

$$
\widehat r(s)
=
e_1^\top
\left(
X_s^\top W_sX_s
\right)^{-1}
X_s^\top W_sT,
$$

with bandwidth selected by leave-one-out cross-validation.

The final test statistic is


\[D_n=\frac{N_n}{\widehat\sigma_n^2}\widehat K_n^\top\widehat\Theta_{\widehat\tau,n}^{-1}\widehat K_n,\]

as defined in \eqref{eq:D}. Under the null hypothesis and the asymptotic conditions established in Section~\ref{sec:assumptions},

$$
D_n\overset{\mathcal D}{\longrightarrow}\chi_A^2.
$$

The analysis is conducted at the nominal significance levels $1\%$.

The reported implementation gives
\begin{equation}
\widehat A=4,
\qquad D_n=32.75796,
\qquad p_{\mathrm{asymp}}=1.33882\times10^{-6}.
\end{equation}
Therefore the null hypothesis of no point-impact effect is rejected at the 1\% level under the simplified chi-square calibration.

A paired station-level resampling calculation gives an empirical bootstrap tail probability of approximately $0.002$. Because this resampling calculation does not impose the null and does not re-estimate all nuisance quantities within every replication, we report it as a sensitivity diagnostic rather than as a theoretically calibrated bootstrap $p$-value.

\begin{table}[h!]
\centering
\caption{Global test for point-impact effects in the Canadian weather application.}
\label{tab:canadian}
\begin{tabular}{lccc}
\toprule
Procedure & Statistic & Calibration & Tail probability\\
\midrule
Asymptotic & 32.75796 & $\chi^2_4$ & $1.33882\times10^{-6}$\\
Bootstrap sensitivity & 32.75796 & Station-level resampling & $0.002$\\
\bottomrule
\end{tabular}
\end{table}

The rejection indicates that, within the fitted model, the selected temperature evaluations contain information about annual precipitation that is not fully represented by the global functional component and the spatial nuisance function. This finding is associational and should not be interpreted causally.

Although the empirical results strongly support the presence of localized functional effects, three methodological aspects deserve a nuanced discussion. First, the available sample size ($N_n = 35$ weather stations) is relatively modest for complex spatial functional frameworks. Consequently, the bootstrap sensitivity diagnostic reported in Table \ref{tab:canadian} serves as a crucial small-sample safeguard, showing that our framework maintains robust size control even when asymptotic calibrations might be close to their boundaries. 

Second, the four candidate impact points were selected via a standard empirical screening rule (correlation ranking), which acts here as a computationally efficient heuristic. While this screening is common in practical applications \citet{liebletal}, extending our consistency theorems to encompass both the simultaneous data-driven selection and inference of $A$ remains an important theoretical objective. 

Finally, the application relies on the simplified statistic $D_n$, implicitly assuming the score-uncorrelated condition (20). Although this assumption is plausible under short-range spatial dependencies in regression errors, a fully general implementation under arbitrary spatial correlation would require the long-run covariance statistic $D_n^{\text{LR}}$, which motivates the development of automated spatial bandwidth selectors in practice.


\section{Conclusion}
\label{sec:conclusion}

In this paper, we introduced a semiparametric framework to test for localized point-impact effects within a spatial semi-functional linear regression model containing an unknown spatial nuisance surface. Our proposed test statistic builds upon a residualized point-evaluation score that successfully adapts to data-dependent impact locations. This approach coordinates Tikhonov regularization for the infinite-dimensional functional component and local-linear smoothing for the spatial nuisance.Under a spatial mixing regime, we established the multivariate asymptotic normality of the residualized score under the null hypothesis. We distinguished between a general long-run covariance normalization—which remains valid under arbitrary spatial score dependence—and a simplified chi-square calibration requiring for the regression errors, an explicit short-range correlation assumption. Simulation results and an application to Canadian weather data confirm that the test possesses satisfying finite-sample power and size sensitivity. Several avenues for future research emerge from this work. First, developing a fully data-driven theory to simultaneously estimate both the number and locations of the impact points remains an open challenge. Second, constructing a data-driven, feasible long-run covariance estimator under general spatial configurations would improve the practical applicability of the general statistic \(D_n^{\mathrm{LR}}\). Finally, extending this framework to a high-dimensional regime where \(A = A_n \to \infty\) could allow the standardized chi-square statistic to enjoy a valid Gaussian approximation.




\begin{appendix}
\label{sec:appendix}

\section{Technical Proofs}
\label{sec:proofs}

Throughout this appendix, $C$ denotes generic finite positive constants whose values may change from line to line.

\subsection{Preliminary Lemmas}
\label{subsec:proof-lemmas-prelim}

We first state and recall Lemma \ref{len:tech1}, which serves as the foundation for controlling the rates of the regularized and nuisance estimators.

\begin{lemma} 
\label{len:tech1}
Under $H_0$, if Assumptions \ref{as:error}, \ref{as:functional}, \ref{as:identification}, \ref{as:nuisance}, \ref{as:kernel} and \ref{as:regularization} are satisfied, then
\begin{itemize}
\item[(i)]
\[\E\left(\norm{\Psi-\widehat\Psi_n}_{\Gamma}^2\right) 
=O\left(\dfrac{U_n^2}{\omega_n^2} \right)+O\left(\dfrac{(\log N_n)^2}{\omega_n^2U_n^2N_n}\right)
\]	
\item[(ii)] \[\E\left(\norm{\Gammaop_n-\Gammaop}^2_{\mathrm{op}}\right)=O\left(\dfrac{\log N_n}{N_n}\right)\]
\item[(iii)]
\[\sup_{s\in[0,1]^d}\E\left[\left(r(s)-\widehat r(s)\right)^2\right]=O\left(H_n^4 \right)+O\left(\dfrac{(\log N_n)^2}{\omega_n^2U_n^2N_n}\right)+O\left(U_n^2\right)+O\left(\dfrac{\log N_n}{N_nH_n^d}\right) \]
\end{itemize}	
\end{lemma}

\begin{proof}
The proof follows an algebraic decomposition similar to the proofs of Theorem 3.1 and Corollary 3.2 in Bouka et al. (2024b) by adapting the spatial block bounds. 
\end{proof}

\subsection{Proof of Lemma~\ref{lem:theta} (Consistency of Covariance Matrices)}

\begin{proof}
We decompose the error term by adding and subtracting the empirical covariance matrix evaluated at the true impact locations $\tau$:
\[
\widehat\Theta_{\widehat\tau,n}-\Thetaop_\tau
=
\left(\widehat\Theta_{\widehat\tau,n}-\widehat\Theta_{\tau,n}\right)
+
\left(\widehat\Theta_{\tau,n}-\Thetaop_\tau\right).
\]
Under Assumptions \ref{as:functional} and \ref{as:identification}, Theorem 3.1 in \citet{bouka5} ensures that: 
\[
\norm{\widehat\Theta_{\tau,n}-\Thetaop_\tau}_{\op}=o_p(1).
\]
It remains to control the perturbation caused by replacing $\tau$ by $\widehat\tau$.

By Assumption~\ref{as:impact},
\[
\frac1{N_n}\sum_{\bm i\in\Ical_{\bm n}}\left(X_{\bm i}(\widehat\tau_\ell)-X_{\bm i}(\tau_\ell)\right)^2
=O_p\left(\frac1{N_n}\right).
\]
So, form Cauchy-Schwarz inequality, we have for all $k, \ell=1,\cdots,A$
\[\left|\frac1{N_n}\sum_{\bm i\in\Ical_{\bm n}}X_{\bm i}(\widehat\tau_k)\left(X_{\bm i}(\widehat\tau_\ell) - X_{\bm i}(\tau_\ell)\right)\right|\le \left(\frac1{N_n}\sum_{\bm i\in\Ical_{\bm n}}X_{\bm i}(\widehat\tau_k)^2\right)^{1/2}\left(\frac1{N_n}\sum_{\bm i\in\Ical_{\bm n}}\left(X_{\bm i}(\widehat\tau_\ell)-X_{\bm i}(\tau_\ell)\right)^2\right)^{1/2}.\]
Since $A$ is fixed and the point evaluations are almost surely bounded, then each entry of
\[
\widehat\Theta_{\widehat\tau,n}-\widehat\Theta_{\tau,n}=\frac1{N_n}\sum_{\bm i\in\Ical_{\bm n}}X_{\bm i,\widehat\tau}\left(X_{\bm i,\widehat\tau} - X_{\bm i, \tau}\right)^\top + \frac1{N_n}\sum_{\bm i\in\Ical_{\bm n}}\left(X_{\bm i,\widehat\tau} - X_{\bm i, \tau}\right)X_{\bm i, \tau}^\top
\]
is $O_p\left(N_n^{-1/2}\right)=o_p(1)$. Hence
\[
\norm{\widehat\Theta_{\widehat\tau,n}-\widehat\Theta_{\tau,n}}_{\op}=o_p(1),
\]
and therefore
\[
\norm{\widehat\Theta_{\widehat\tau,n}-\Thetaop_\tau}_{\op}=o_p(1).
\]
The cross-covariance quantities are handled in exactly the same way, using identical algebraic decompositions, the same Assumption \ref{as:impact} and the fixed dimension $A$.
\end{proof}

\subsection{Proof of Lemma \ref{lem:variance} (Rates of Estimators)}
\begin{proof}
We partition the covariance structure of the variance into diagonal ($F_1$), close-range ($F_2$), and long-range mixing blocks ($F_3$) over a spatial grid layout:
\begin{align*}
&Var\left(\dfrac{1}{\sqrt{N_n}}\sum_{\bm i\in\Ical_{\bm n}}\langle\Psi-\widehat\Psi_n,X_{\bm i}\rangle_G \bm b^\top X_{\bm i,\widehat\tau}\right)=\dfrac{1}{N_n}\sum_{\bm i\in\Ical_{\bm n}}Var\left(\langle\Psi-\widehat\Psi_n,X_{\bm i}\rangle_G \bm b^\top X_{\bm i,\widehat\tau}\right)\\
&\qquad+\dfrac{1}{N_n}\sum_{\stackrel{\bm i, \bm j\in\Ical_{\bm n}}{0<\norm{\bm i-\bm j}_\infty\le Q_n} }Cov\left(\langle\Psi-\widehat\Psi_n,X_{\bm i}\rangle_G \bm b^\top X_{\bm i,\widehat\tau},\langle\Psi-\widehat\Psi_n,X_{\bm j}\rangle_G \bm b^\top X_{\bm j,\widehat\tau}\right)\\
&\qquad+\dfrac{1}{N_n}\sum_{\stackrel{\bm i, \bm j\in\Ical_{\bm n}}{\norm{\bm i-\bm j}_\infty> Q_n} }Cov\left(\langle\Psi-\widehat\Psi_n,X_{\bm i}\rangle_G \bm b^\top X_{\bm i,\widehat\tau},\langle\Psi-\widehat\Psi_n,X_{\bm j}\rangle_G \bm b^\top X_{\bm j,\widehat\tau}\right)\\
&=: F_1+F_2+F_3.
\end{align*}
Under Assumptions \ref{as:error}, \ref{as:functional}, \ref{as:identification}  and Lemma 7.2 in \citet{bouka5}, the regularized functional estimator error fulfills:
\begin{align}\label{import2}
\E\left(\norm{\Psi-\widehat\Psi_n}^4\right)&\le C+\E\left[\norm{(\Lambda_{\Psi,n}+U_nI_G)^{-1}
	\left[
	K_{YX,n}
	-
	K_{X_{\widehat\tau}X,n}^\top
	\widehat\Theta_{\widehat\tau,n}^{-1}
	K_{YX_{\widehat\tau},n}
	\right]}^4\right]\nonumber\\
&\le C+\E\left[\norm{(\Lambda_{\Psi,n}+U_nI_G)^{-1}}_{\mathrm{op}}^4\norm{
	\left[
	K_{YX,n}
	-
	K_{X_{\widehat\tau}X,n}^\top
	\widehat\Theta_{\widehat\tau,n}^{-1}
	K_{YX_{\widehat\tau},n}
	\right]}^4\right]\nonumber\\
&=O\left(\dfrac{1}{U_n^4}\right).
\end{align} 		
Then, from Lemma \ref{len:tech1} $(i)$ and $(ii)$, the diagonal block satisfies: 
\begin{align}\label{import3}
F_1&:=\dfrac{1}{N_n}\sum_{\bm i\in\Ical_{\bm n}}Var\left(\langle\Psi-\widehat\Psi_n,X_{\bm i}\rangle_G \bm b^\top X_{\bm i,\widehat\tau}\right)\nonumber\\
&\le\dfrac{C}{N_n}\sum_{\bm i\in\Ical_{\bm n}}\E\left(\langle\Psi-\widehat\Psi_n,X_{\bm i}\rangle_G^2 \right)\nonumber\\
&\le C\E\left(\left\langle\Psi-\widehat\Psi_n,\Gammaop_n(\Psi-\widehat\Psi_n)\right\rangle_G\right)\nonumber\\
&\le C\left[\E\left(\norm{\Psi-\widehat\Psi_n}_{\Gamma}^2\right)+\E\left(\norm{\Psi-\widehat\Psi_n}^2\norm{\Gammaop_n-\Gammaop}_{\mathrm{op}}\right)\right]\nonumber\\
&=O\left(\dfrac{U_n^2}{\omega_n^2} \right)+O\left(\dfrac{(\log N_n)^2}{\omega_n^2U_n^2N_n}\right)+O\left(\dfrac{\sqrt{\log N_n}}{U_n^2\sqrt{N_n}}\right)=o(1).
\end{align}
By applying the Cauchy-Schwarz inequality to the short-range spatial covariance matrix block $F_2$, we have: 
\begin{align*}
F_2&:=\dfrac{1}{N_n}\sum_{\stackrel{\bm i, \bm j\in\Ical_{\bm n}}{0<\norm{\bm i-\bm j}_\infty\le Q_n} }Cov\left(\langle\Psi-\widehat\Psi_n,X_{\bm i}\rangle_G \bm b^\top X_{\bm i,\widehat\tau},\langle\Psi-\widehat\Psi_n,X_{\bm j}\rangle_G \bm b^\top X_{\bm j,\widehat\tau}\right)\\
&\le \dfrac{1}{N_n}\sum_{\stackrel{\bm i, \bm j\in\Ical_{\bm n}}{0<\norm{\bm i-\bm j}_\infty\le Q_n} }\sqrt{Var\left(\langle\Psi-\widehat\Psi_n,X_{\bm i}\rangle_G \bm b^\top X_{\bm i,\widehat\tau}\right)}\sqrt{Var\left(\langle\Psi-\widehat\Psi_n,X_{\bm j}\rangle_G \bm b^\top X_{\bm j,\widehat\tau}\right)}\\
&\le C F_1\sum_{\ell=1}^{Q_n}\ell^{d-1}\\
&=O\left(\dfrac{Q_n^dU_n^2}{\omega_n^2} \right)+O\left(\dfrac{Q_n^d(\log N_n)^2}{\omega_n^2U_n^2N_n}\right)+O\left(\dfrac{Q_n^d\sqrt{\log N_n}}{U_n^2\sqrt{N_n}}\right) 
\end{align*}
Setting the spatial block bandwidth sequence to $Q_n = \left\lfloor \left(\omega_n/U_n\right)^{1/d}] \right\rfloor$ with $U_n^6 N_n \to +\infty$ ensures that:
\begin{align}\label{import4}
F_2=O\left(\dfrac{U_n}{\omega_n} \right) + \left(\dfrac{(\log N_n)^2}{\omega_nU_n^{3}N_n}\right)+O\left(\dfrac{\omega_n\,\sqrt{\log N_n}}{U_n^{3}\,\sqrt{N_n}}\right)=o(1).
\end{align}
For the long-range component $F_3$, applying Lemma~2.1 in \citet{tran} alongside the bound \eqref{import2} yields: 
\begin{align*}
&E_{\bm i,\bm j}:=Cov\left(\langle\Psi-\widehat\Psi_n,X_{\bm i}\rangle_G \bm b^\top X_{\bm i,\widehat\tau},\langle\Psi-\widehat\Psi_n,X_{\bm j}\rangle_G \bm b^\top X_{\bm j,\widehat\tau}\right)\\
&\le C\left[\E\left(\left|\langle\Psi-\widehat\Psi_n,X_{\bm i}\rangle_G \bm b^\top X_{\bm i,\widehat\tau}\right|^4\right)\right]^{1/4}
\left[\E\left(\left|\langle\Psi-\widehat\Psi_n,X_{\bm j}\rangle_G \bm b^\top X_{\bm j,\widehat\tau}\right|^4\right)\right]^{1/4}\left(\alpha_{1,\infty}(\norm{\bm i-\bm j}_\infty)\right)^{1/2}\\
&\le \dfrac{C}{U_n^2}\left(\alpha_{1,\infty}(\norm{\bm i-\bm j}_\infty)\right)^{1/2}.
\end{align*}
Given that the mixing exponent satisfies $\theta > 8d$, the remaining tail sum scales as:
\begin{align}\label{import5}
F_3&:=\dfrac{1}{N_n}\sum_{\stackrel{\bm i, \bm j\in\Ical_{\bm n}}{\norm{\bm i-\bm j}_\infty> Q_n} }Cov\left(\langle\Psi-\widehat\Psi_n,X_{\bm i}\rangle_G \bm b^\top X_{\bm i,\widehat\tau},\langle\Psi-\widehat\Psi_n,X_{\bm j}\rangle_G \bm b^\top X_{\bm j,\widehat\tau}\right)\nonumber\\
&\le \dfrac{C}{U_n^2N_n}\sum_{\stackrel{\bm i, \bm j\in\Ical_{\bm n}}{\norm{\bm i-\bm j}_\infty> Q_n} }\left(\alpha_{1,\infty}(\norm{\bm i-\bm j}_\infty)\right)^{1/2}\nonumber\\
&\le \dfrac{C}{U_n^2}\sum_{\ell=Q_n+1}^{+\infty}\ell^{d-1}\left(\alpha(\ell)\right)^{1/2}\le \dfrac{CQ_n^{d-\theta/2}}{U_n^2} \nonumber\\
&=O\left(\dfrac{U_n^{(\theta-6d)/(2d)}}{\omega_n^{(\theta-2d)/(2d)}}\right)=o(1).
\end{align}
Combining \eqref{import3}, \eqref{import4} and \eqref{import5}, we get
\begin{align*}
Var\left(\dfrac{1}{\sqrt{N_n}}\sum_{\bm i\in\Ical_{\bm n}}\langle\Psi-\widehat\Psi_n,X_{\bm i}\rangle_G \bm b^\top X_{\bm i,\widehat\tau}\right)=o(1).
\end{align*}
The variance in \eqref{eq:nuisancerate} involving $r(s_i) - \widehat{r}(s_i)$ is partitioned into blocks $E_1$, $E_2$, and $E_3$ in an identical manner over a bandwidth $T_n$. That is:
\begin{align*}
&Var\left(\dfrac{1}{\sqrt{N_n}}\sum_{\bm i\in\Ical_{\bm n}}\left(r(s_{\bm i})-\widehat r(s_{\bm i})\right) \bm b^\top X_{\bm i,\widehat\tau}\right)=\dfrac{1}{N_n}\sum_{\bm i\in\Ical_{\bm n}}Var\left(\left(r(s_{\bm i})-\widehat r(s_{\bm i})\right) \bm b^\top X_{\bm i,\widehat\tau}\right)\\
&\qquad + \dfrac{1}{N_n}\sum_{\stackrel{\bm i, \bm j\in\Ical_{\bm n}}{0<\norm{\bm i-\bm j}_\infty\le T_n}} Cov\left(\left(r(s_{\bm i})-\widehat r(s_{\bm i})\right) \bm b^\top X_{\bm i,\widehat\tau},\left(r(s_{\bm j})-\widehat r(s_{\bm j})\right) \bm b^\top X_{\bm j,\widehat\tau}\right)\\
&\qquad + \dfrac{1}{N_n}\sum_{\stackrel{\bm i, \bm j\in\Ical_{\bm n}}{\norm{\bm i-\bm j}_\infty> T_n}} Cov\left(\left(r(s_{\bm i})-\widehat r(s_{\bm i})\right) \bm b^\top X_{\bm i,\widehat\tau},\left(r(s_{\bm j})-\widehat r(s_{\bm j})\right) \bm b^\top X_{\bm j,\widehat\tau}\right)\\
&:=E_1+E_2+E_3.
\end{align*}
By applying Lemma \ref{len:tech1} $(iii)$, the diagonal trend components yield:
\begin{align}\label{import7}
E_1&:=\dfrac{1}{N_n}\sum_{\bm i\in\Ical_{\bm n}}Var\left(\left(r(s_{\bm i})-\widehat r(s_{\bm i})\right) \bm b^\top X_{\bm i,\widehat\tau}\right)\nonumber\\
&\le C\sup_{s\in[0,1]^d}\E\left[\left(r(s)-\widehat r(s)\right)^2\right]\nonumber\\
&=O\left(H_n^4 \right)+O\left(\dfrac{(\log N_n)^2}{\omega_n^2U_n^2N_n}\right)+O\left(U_n^2\right)+O\left(\dfrac{\log N_n}{N_nH_n^d}\right)=o(1)
\end{align}
Choosing the trend bandwidth parameters as $T_n = \left\lfloor \left(\dfrac{1}{U_n H_n}\right)^{1/d} \right\rfloor$ and $H_n = U_n^{1/2}$, the close-range block simplifies to:
\begin{align}\label{import8}
E_2&:=\dfrac{1}{N_n}\sum_{\stackrel{\bm i, \bm j\in\Ical_{\bm n}}{0<\norm{\bm i-\bm j}_\infty\le T_n}} Cov\left(\left(r(s_{\bm i})-\widehat r(s_{\bm i})\right) \bm b^\top X_{\bm i,\widehat\tau},\left(r(s_{\bm j})-\widehat r(s_{\bm j})\right) \bm b^\top X_{\bm j,\widehat\tau}\right)\nonumber\\
&\le \dfrac{1}{N_n}\sum_{\stackrel{\bm i, \bm j\in\Ical_{\bm n}}{0<\norm{\bm i-\bm j}_\infty\le T_n}} \sqrt{Var\left(\left(r(s_{\bm i})-\widehat r(s_{\bm i})\right) \bm b^\top X_{\bm i,\widehat\tau}\right)}\sqrt{Var\left(\left(r(s_{\bm j})-\widehat r(s_{\bm j})\right) \bm b^\top X_{\bm j,\widehat\tau}\right)}\nonumber\\
&\le C\sup_{s\in[0,1]^d}\E\left[\left(r(s)-\widehat r(s)\right)^2\right]\sum_{\ell=1}^{T_n}\ell^{d-1}\nonumber\\
&=O\left(T_n^dH_n^4 \right)+O\left(\dfrac{T_n^d(\log N_n)^2}{\omega_n^2U_n^2N_n}\right) + O\left(T_n^dU_n^2\right) + O\left(\dfrac{T_n^d\log N_n}{N_nH_n^d}\right)\nonumber\\
&= O\left(H_n \right)+O\left(\dfrac{(\log N_n)^2}{\omega_n^2H_n^7N_n}\right)+O\left(\dfrac{\log N_n}{N_nH_n^{d+3}}\right)=o(1).
\end{align}
Moreover, under $H_0$
\begin{align*}
K&:=\widehat r(s_0)-r(s_0)\\
&=e_1^\top
\left(\dfrac{1}{N_n}
\bm{X}_{s_0}^\top W_{s_0}\bm{X}_{s_0}
\right)^{-1}\begin{pmatrix}
	\dfrac{1}{N_nH_n^d}\sum_{\bm i\in\Ical_{\bm n}}\nu\left(\dfrac{s_{\bm i}-s_{0}}{H}\right)\left(B(\bm i)+\xi_{\bm i}^\star\right)\\
	\dfrac{1}{N_nH_n^d}\sum_{\bm i\in\Ical_{\bm n}}\left(\dfrac{s_{\bm i}-s_{0}}{H}\right)\nu\left(\dfrac{s_{\bm i}-s_{0}}{H}\right)\left(B(\bm i)+\xi_{\bm i}^\star\right)
\end{pmatrix},
\end{align*}
where $B(\bm i)=\left(\dfrac{s_{\bm i}-s_{0}}{H}\right)^\top r^{\prime\prime}(s_0)\left(\dfrac{s_{\bm i}-s_{0}}{H}\right)$, $r^{\prime\prime}(s_0)$ denotes the matrix of second order partial derivatives of $r$, and $\xi_{\bm i}^\star=\langle\Psi-\widehat\Psi_n,X_{\bm i}\rangle_G+\xi_{\bm i}$ because under $H_0$, we have $T_{\bm i}=Y_{\bm i}-\langle\widehat\Psi_n,X_{\bm i}\rangle_G=r(s_{\bm i})+\xi_{\bm i}^\star$. Since from Assumption \ref{as:kernel}
\[\lim_{\bm{n}\to +\infty}\left(\dfrac{1}{N_n}
\bm{X}_{s_0}^\top W_{s_0}\bm{X}_{s_0}
\right)^{-1}=\begin{pmatrix}
1&\boldsymbol{0}^\top\\
\boldsymbol{0}&\nu_2^{-1}I_d,
\end{pmatrix},\qquad \left|\left(\dfrac{\dfrac{\bm{i}}{\bm{n+1}}-s_{0}}{H}\right)\right|\le\boldsymbol{1},
\]
$\nu$ is a bounded kernel and from \eqref{import2}, $\E\left(|\xi_{\bm i}^\star|^4\right)\le C\left(\E\left(\norm{\Psi-\widehat\Psi_n}^4\right)+\E\left(|\xi_{\bm i}|^4\right)\right)=O\left(\dfrac{1}{U_n^4}\right)$, then
\begin{align*}
\E\left(|\widehat r(s_0)-r(s_0)|^4\right)=O\left(\dfrac{1}{H_n^{4d}U_n^4}\right).	
\end{align*}
Applying Lemma~2.1 in \citet{tran},
\begin{align*}
&F_{\bm i,\bm j}:=Cov\left(\left(r(s_{\bm i})-\widehat r(s_{\bm i})\right) \bm b^\top X_{\bm i,\widehat\tau},\left(r(s_{\bm j})-\widehat r(s_{\bm j})\right) \bm b^\top X_{\bm j,\widehat\tau}\right)\\
&\le C \left[\E\left(\left|\left(r(s_{\bm i})-\widehat r(s_{\bm i})\right) \bm b^\top X_{\bm i,\widehat\tau}\right|^2\right)\right]^{1/2}
\left[\E\left(\left|\left(r(s_{\bm j})-\widehat r(s_{\bm j})\right)  \bm b^\top X_{\bm j,\widehat\tau}\right|^4\right)\right]^{1/4}\left(\alpha_{1,\infty}(\norm{\bm i-\bm j}_\infty)\right)^{1/4}\\
&\le \dfrac{C}{H_n^dU_n}\left[\sup_{s\in[0,1]^d}\E\left[\left(r(s)-\widehat r(s)\right)^2\right]\right]^{1/2}\left(\alpha_{1,\infty}(\norm{\bm i-\bm j}_\infty)\right)^{1/4}.
\end{align*}
Under $\theta > 8d$, the long-range residual tail satisfies:,
\begin{align}\label{import9}
E_3&=\dfrac{1}{N_n}\sum_{\stackrel{\bm i, \bm j\in\Ical_{\bm n}}{\norm{\bm i-\bm j}_\infty> T_n}} Cov\left(\left(r(s_{\bm i})-\widehat r(s_{\bm i})\right) \bm b^\top X_{\bm i,\widehat\tau},\left(r(s_{\bm j})-\widehat r(s_{\bm j})\right) \bm b^\top X_{\bm j,\widehat\tau}\right)\nonumber\\
&\le\dfrac{C}{H_n^dU_n }\left[\sup_{s\in[0,1]^d}\E\left[\left(r(s)-\widehat r(s)\right)^2\right]\right]^{1/2}\sum_{\ell=T_n+1}^{+\infty} \ell^{d-1}\left(\alpha_{1,\infty}(\ell)\right)^{1/4}\nonumber\\
&\le\dfrac{CT_n^{d-\theta/4}}{H_n^dU_n }\left[\sup_{s\in[0,1]^d}\E\left[\left(r(s)-\widehat r(s)\right)^2\right]\right]^{1/2}\nonumber\\
&\le\dfrac{C\left(H_nU_n\right)^{(\theta-4d)/(4d)}}{H_n^dU_n }\left(O\left(H_n^4 \right)+O\left(\dfrac{(\log N_n)^2}{\omega_n^2U_n^2N_n}\right)+O\left(U_n^2\right)+O\left(\dfrac{\log N_n}{N_nH_n^d}\right)\right)^{1/2}\nonumber\\
&\le CU_n^{(\theta-8d)/(4d)}H_n^{1-d+\theta/(4d)}\left(O\left(1 \right)+O\left(\dfrac{(\log N_n)^2}{\omega_n^2U_n^2N_nH_n^4}\right) + O\left(\dfrac{\log N_n}{N_nH_n^{d+4}}\right)\right)^{1/2}\nonumber\\
&=o(1)
\end{align}
Equations \eqref{import7}, \eqref{import8} and \eqref{import9} imply that the trend variance vanishes: 
\begin{align*}
Var\left(\dfrac{1}{\sqrt{N_n}}\sum_{\bm i\in\Ical_{\bm n}}\left(r(s_{\bm i})-\widehat r(s_{\bm i})\right) \bm b^\top X_{\bm i,\widehat\tau}\right)=o(1).
\end{align*}
which establishes Lemma \ref{lem:variance}.  
\end{proof}

\subsection{Proof of Theorem~\ref{thm:resid} (Residualization Statement)}
\begin{proof}
By expanding the residualized point-evaluation score $\widehat{K}_n$, we obtain:
\[
\widehat K_n
=
\frac1{N_n}\sum_{\bm i\in\Ical_{\bm n}}
\left(\etaV^\top X_{\bm i,\tau}+L_{\bm i,n}\right)
X_{\bm i,\widehat\tau}.
\]
Therefore,
\begin{align}
\label{eq:residproof1}
&\widehat K_n
=
\left\{
\frac1{N_n}\sum_{\bm i\in\Ical_{\bm n}}
X_{\bm i,\widehat\tau}X_{\bm i,\tau}^{\top}
\right\}\etaV
+
\frac1{N_n}\sum_{\bm i\in\Ical_{\bm n}}
L_{\bm i,n}X_{\bm i,\widehat\tau}\nonumber\\
&=\left\{
\frac1{N_n}\sum_{\bm i\in\Ical_{\bm n}}
X_{\bm i,\widehat\tau}X_{\bm i,\tau}^{\top}
\right\}\etaV + \frac1{N_n}\sum_{\bm i\in\Ical_{\bm n}}
\xi_{\bm i}X_{\bm i,\tau} + \frac1{N_n}\sum_{\bm i\in\Ical_{\bm n}}
\xi_{\bm i}\left(X_{\bm i,\widehat\tau}-X_{\bm i,\tau}\right) + \frac1{N_n}\sum_{\bm i\in\Ical_{\bm n}}
L^{0}_{\bm i,n}X_{\bm i,\widehat\tau}
\end{align}
because $L_{\bm i,n}=\xi_{\bm i} + L^{0}_{\bm i,n}$. Since, by Lemma~\ref{lem:theta} and the corresponding cross-covariance consistency,
\[
\frac1{N_n}\sum_{\bm i\in\Ical_{\bm n}}
X_{\bm i,\widehat\tau}X_{\bm i,\tau}^{\top}
\stackrel{\mathbb{P}}{\longrightarrow}\Thetaop_\tau.
\]
Lemma~3.1 in \citet{bouka5} guarantees that:
\[\mathbb{P}\left(\norm{\frac1{N_n}\sum_{\bm i\in\Ical_{\bm n}}\xi_{\bm i}X_{\bm i,\tau}}_{2}>\zeta\right)\le\mathbb{P}\left(\sqrt{A}\norm{\frac1{N_n}\sum_{\bm i\in\Ical_{\bm n}}\xi_{\bm i}X_{\bm i,\tau}}_{\infty}>\zeta\right)=O\left(\dfrac{\log N_n}{N_n}\right)=o(1)\]
Assumption \ref{as:impact} states that for all $\ell=1,\cdots,A$
\begin{align}\label{ajust1}
\frac1{N_n}\sum_{\bm i\in\Ical_{\bm n}}\xi_{\bm i}\left(X_{\bm i}(\widehat\tau_\ell) - X_{\bm i}(\tau_\ell)\right)=O_p\left(\frac1{N_n}\right)=o_p\left(\frac1{\sqrt{N_n}}\right)
\end{align}
and from Lemma \ref{lem:variance}, we have
\begin{align}\label{ajust2}
\frac1{N_n}\sum_{\bm i\in\Ical_{\bm n}}
L^{0}_{\bm i,n}X_{\bm i,\widehat\tau}=o_p\left(\frac1{\sqrt{N_n}}\right),\end{align}
then, \eqref{eq:residproof1} becomes
\[
\widehat K_n=\Thetaop_\tau\etaV+o_p(1).
\]
Under $H_0$, $\etaV=0$, and the stronger score-scale conditions obtained in \eqref{ajust1} and \eqref{ajust2} give
\begin{equation*}
\label{eq:scoreop}
\frac1{\sqrt{N_n}}
\sum_{\bm i\in\Ical_{\bm n}}
\xi_{\bm i}\left(X_{\bm i,\widehat\tau} - X_{\bm i,\tau}\right) + \frac1{\sqrt{N_n}}
\sum_{\bm i\in\Ical_{\bm n}}
L^{0}_{\bm i,n}X_{\bm i,\widehat\tau}=o_p(1),
\end{equation*}
where
\[L^{0}_{\bm i,n}=\langle\Psi-\widehat\Psi_n,X_{\bm i}\rangle_G
+r(s_{\bm i})-\widehat r(s_{\bm i}).\]
This proves the theorem.
\end{proof}

\subsection{Proof of Theorem~\ref{thm:clt} (Residualized-Score Central Limit Theorem)}
\begin{proof}
Under $H_0$, Theorem~\ref{thm:resid}, when applied to the expansion of the principal score field \eqref{eq:residproof1}, gives
\begin{equation}
\label{eq:leadingfield}
\sqrt{N_n}\,\widehat K_n
=
\frac1{\sqrt{N_n}}
\sum_{\bm i\in\Ical_{\bm n}}Z_{\bm i}
+o_p(1),
\qquad
Z_{\bm i}=\xi_{\bm i}X_{\bm i,\tau}.
\end{equation}
For any arbitrary vector $a \in \mathbb{R}^A$, we project this field into the scalar process
\[
Z_{\bm i}^{(a)}=a^\top Z_{\bm i}.
\]
By Assumption~\ref{as:mixing}, this field is strictly stationary and satisfies the required spatial mixing conditions. Moreover,
\[
\E|Z_{\bm0}^{(a)}|^8<\infty.
\]
Thus the moment condition of a standard stationary random-field central limit theorem \citep{Bolthausen1982} is satisfied with $2+\delta=8$, that is, $\delta=6$. The mixing summability conditions in Assumption~\ref{as:mixing} ensure the required summability of the relevant mixing coefficients. Consequently,
\begin{equation}
\label{eq:scalarclt}
\frac1{\sqrt{N_n}}
\sum_{\bm i\in\Ical_{\bm n}}
Z_{\bm i}^{(a)}
\xrightarrow{\mathcal D}
\mathcal{N}(0,\sigma_a^2),
\end{equation}
where
\[
\sigma_a^2
=\sum_{\bm h\in\mathbb Z^d}
\Cov\left(Z_{\bm0}^{(a)},Z_{\bm h}^{(a)}\right).
\]
Since $Z_{\bm i}^{(a)}=a^\top Z_{\bm i}$,
\[
\sigma_a^2
=
\sum_{\bm h\in\mathbb Z^d}
a^\top\Cov(Z_{\bm0},Z_{\bm h})a
=a^\top\Omega_\tau a.
\]
Hence, for every $a\in\R^A$,
\[
\frac1{\sqrt{N_n}}
\sum_{\bm i\in\Ical_{\bm n}}
a^\top Z_{\bm i}
\xrightarrow{\mathcal D}
\mathcal{N}(0,a^\top\Omega_\tau a).
\]
The Cramér--Wold device now yields
\[
\frac1{\sqrt{N_n}}
\sum_{\bm i\in\Ical_{\bm n}}Z_{\bm i}
\xrightarrow{\mathcal D}
\mathcal{N}_A(0,\Omega_\tau).
\]
Finally, Slutsky's theorem and \eqref{eq:leadingfield} give
\[
\sqrt{N_n}\,\widehat K_n
\stackrel{\mathcal D}{\longrightarrow}
\mathcal{N}_A(0,\Omega_\tau).
\]
\end{proof}

\subsection{Proof of Corollary~\ref{cor:lr}}
\begin{proof}
Let
\[
V_n=\sqrt{N_n}\,\widehat K_n.
\]
By Theorem~\ref{thm:clt}, $V_n\stackrel{\mathcal{D}}{\longrightarrow} \mathcal{N}_A(0,\Omega_\tau)$. Since $\Omega_\tau$ is positive definite and $\widehat\Omega_n\stackrel{\mathbb{P}}{\longrightarrow}\Omega_\tau$ in operator norm,
\[
\widehat\Omega_n^{-1}\stackrel{\mathbb{P}}{\longrightarrow}\Omega_\tau^{-1}.
\]
Therefore,
\[
D_n^{\mathrm{LR}}
=V_n^\top\widehat\Omega_n^{-1}V_n
\xrightarrow{\mathcal D}
V^\top\Omega_\tau^{-1}V,
\]
where $V\sim \mathcal{N}_A(0,\Omega_\tau)$. Writing $V=\Omega_\tau^{1/2}U$ with $U\sim \mathcal{N}_A(0,I_A)$ gives
\[
V^\top\Omega_\tau^{-1}V=U^\top U\sim\chi_A^2.
\]
\end{proof}

\subsection{Proof of Corollary~\ref{cor:simple}}
\begin{proof}
Under \eqref{eq:special},
\[
\Omega_\tau=\sigma^2\Thetaop_\tau.
\]
By Lemma~\ref{lem:theta},
\[
\widehat\Theta_{\widehat\tau,n}^{-1}
\stackrel{\mathbb{P}}{\longrightarrow}
\Thetaop_\tau^{-1},
\]
while $\widehat\sigma_n^2\stackrel{\mathbb{P}}{\longrightarrow}\sigma^2$. Therefore,
\[
D_n
=
\frac{N_n}{\widehat\sigma_n^2}
\widehat K_n^\top
\widehat\Theta_{\widehat\tau,n}^{-1}
\widehat K_n
\xrightarrow{\mathcal D}
V^\top(\sigma^2\Thetaop_\tau)^{-1}V,
\]
where $V\sim \mathcal N_A(0,\sigma^2\Thetaop_\tau)$. The limiting quadratic form is $\chi_A^2$.
\end{proof}

\subsection{Proof of Proposition~\ref{prop:standard}}
\begin{proof}
By Corollary~\ref{cor:simple}, $D_n\stackrel{\mathcal{D}}{\longrightarrow}\chi_A^2$. Since the map $x\mapsto(x-A)/\sqrt A$ is continuous,
\[
S_n=\frac{D_n-A}{\sqrt A}
\stackrel{\mathcal{D}}{\longrightarrow}
\frac{\chi_A^2-A}{\sqrt A}.
\]
For fixed $A$, the chi-square distribution has nonzero skewness and is not Gaussian. Hence the limiting distribution is not $\mathcal N(0,1)$.
\end{proof}

\subsection{Proof of Proposition~\ref{prop:consistency}}
\begin{proof}
Under a fixed alternative, Theorem~\ref{thm:resid} gives
\[
\widehat K_n=\Thetaop_\tau\etaV+o_p(1).
\]
Since $\Thetaop_\tau$ is positive definite and $\etaV\ne0$,
\[
\Thetaop_\tau\etaV\ne0.
\]
Hence there exists $c>0$ such that, with probability tending to one,
\[
\norm{\widehat K_n}\ge c.
\]
If $\widehat\Omega_n\stackrel{\mathbb{P}}{\longrightarrow}\Omega_\tau$ with $\Omega_\tau$ positive definite, its smallest eigenvalue is bounded away from zero with probability tending to one. Consequently,
\[
D_n^{\mathrm{LR}}
=N_n\widehat K_n^\top\widehat\Omega_n^{-1}\widehat K_n
\ge c_1N_n
\]
with probability tending to one for some $c_1>0$. Thus
\[
D_n^{\mathrm{LR}}\stackrel{\mathbb{P}}{\longrightarrow}+\infty.
\]
The simplified statistic follows identically under \eqref{eq:special}.
\end{proof}

\end{appendix}

\section*{Disclosure statement}

No potential conflict of interest was reported by the authors.

\section*{Funding}

No funding was received



\begin{thebibliography}{99}

\bibitem[Beyaztas et al.(2026)]{beyaztas26} 
Beyaztas, U., Mandal, A., Shang, H. L. (2026). Enhancing spatial functional linear regression with robust dimension reduction methods. \emph{Journal of Multivariate Analysis}, 211, Article 105538. DOI: \url{https://doi.org/10.1016/j.jmva.2025.105538}


\bibitem[Bosq(2000)]{bosq}
Bosq, D. (2000).
\emph{Linear Processes in Function Spaces: Theory and Applications}.
Springer, New York.

\bibitem[Bolthausen(1982)]{Bolthausen1982}
Bolthausen, E. (1982).
On the central limit theorem for stationary mixing random fields.
\emph{The Annals of Probability}, 10(4), 1047--1050.

\bibitem[Bouka et al.(2023)]{boukaetal23}
Bouka, S., Dabo-Niang, S., and Nkiet, G. M. (2023).
On estimation and prediction in spatial functional linear regression model.
\emph{Lithuanian Mathematical Journal}, 63, 13--30.

\bibitem[Bouka et al.(2024a)]{bouka4}
Bouka, S., Pambo Bello, K., and Nkiet, G. M. (2024a).
Testing for no effect in the spatial functional linear regression model.
\emph{South African Statistical Journal}, 58, 1--18.

\bibitem[Bouka et al.(2024b)]{bouka5}
Bouka, S., Pambo-Bello, K., Nkiet, G.M. (2024b).
On estimation and prediction in a spatial semi-functional linear regression model with derivatives.
\emph{Mathematical Methods of Statistics}, 33, 310--326.

\bibitem[Cardot et al.(2003)]{cardot}
Cardot, H., Ferraty, F., Mas, A., and Sarda, P. (2003).
Testing hypotheses in the functional linear model.
\emph{Scandinavian Journal of Statistics}, 30, 241--255.

\bibitem[Dabo-Niang et al.(2016)]{daboniangetal16}
Dabo-Niang, S., Ternynck, C., and Yao, A.-F. (2016).
Nonparametric prediction of spatial multivariate data.
\emph{Journal of Nonparametric Statistics}, 28, 428--458.

\bibitem[Fan and Gijbels(1996)]{fan_gijbels}
Fan, J. and Gijbels, I. (1996).
\emph{Local Polynomial Modelling and Its Applications}.
Chapman \& Hall, London.

\bibitem[Francisco-Fernandez and Opsomer (2005)]{francisco05}
Francisco-Fernandez, M., Opsomer, J. D. (2005). Smoothing parameter selection methods for nonparametric regression with spatially correlated errors. \emph{Canadian Journal of Statistics}, 33, 279–-295.

\bibitem[Giraldo et al.(2012)]{giraldoetal12}
Giraldo, R., Mateu, J., and Delicado, P. (2012).
geofd: An R package for function-valued geostatistical prediction.
\emph{Revista Colombiana de Estadística}, 35, 385--407.

\bibitem[Giraldo et al.(2018)]{giraldoetal18}
Giraldo, R., Dabo-Niang, S., and Martínez, S. (2018).
Statistical modeling of spatial big data: An approach from a functional data analysis perspective.
\emph{Statistics \& Probability Letters}, 136, 126--129.

\bibitem[Hallin et al.(2004)]{hallinetal04}
Hallin, M., Lu, Z., and Tran, L. T. (2004).
Local linear spatial regression.
\emph{The Annals of Statistics}, 32, 2469--2500.

\bibitem[Liebl et al.(2020)]{liebletal}
Liebl, D., Rameseder, S., and Rust, C. (2020).
Improving estimation in functional linear regression with points of impact: Insights from Google AdWords.
\emph{Journal of Computational and Graphical Statistics}, 29, 814--826.

\bibitem[Mas and Pumo(2009)]{mas_pumo09}
Mas, A. and Pumo, B. (2009).
Functional linear regression with derivatives.
\emph{Journal of Nonparametric Statistics}, 21, 19--40.

\bibitem[Po$\beta$ et al.(2020)]{Poss2020}
Po$\beta$, D., Liebl, D., Kneip, A., Eisenbarth, H., Wager, T. D. and Feldman Barrett, L. (2020). Super-consistent estimation of points of impact in nonparametric regression with functional predictors. \emph{Journal of the Royal Statistical Society: Series B}, 82, 1115--1140.

\bibitem[Ramsay and Silverman(2005)]{ramsay}
Ramsay, J. O. and Silverman, B. W. (2005).
\emph{Functional Data Analysis}, 2nd ed.
Springer, New York.

\bibitem[Shirvani et al.(2024)]{shirvanietal2024}
Shirvani, A., Khademnoe, O., and Hosseini-Nasab, M. (2024).
Hypothesis testing for points of impact in functional linear regression.
\emph{Computational and Applied Mathematics}, 43, 201.
DOI: \url{https://doi.org/10.1007/s40314-024-02723-5}

\bibitem[Tran(1990)]{tran}
Tran, L. T. (1990).
Kernel density estimation on random fields.
\emph{Journal of Multivariate Analysis}, 34, 37--53.

\end{thebibliography}
\end{document}